\documentclass[runningheads]{llncs}
\usepackage{amsmath, bm, amssymb}
\usepackage{hyperref}
\usepackage[capitalise, noabbrev]{cleveref}
\renewcommand{\emph}[1]{\textit{\textbf{#1}}}
\usepackage[utf8]{inputenc}
\usepackage[nobreak]{mdframed}
\usepackage{cite}

\usepackage{xcolor}
\definecolor{blue}{rgb}{0.032812499999999994, 0.0390625, 0.95390624999999996}
\hypersetup{colorlinks = true, citecolor = blue, linkcolor = blue}

\usepackage[ruled, titlenotnumbered, noline]{algorithm2e} % linesnumbered
\crefname{algocf}{Algorithm}{Algorithms}
\Crefname{algocf}{Algorithm}{Algorithms}

\SetCommentSty{commentfont}
\SetKwComment{complexity}{(}{)}

\usepackage{thmtools}
\usepackage{thm-restate}

\title{Adaptive Exact Learning in a Mixed-Up World:
Dealing with Periodicity, Errors and Jumbled-Index Queries in String Reconstruction}

\titlerunning{Adaptive Exact Learning in a Mixed-Up World}

\author{Ramtin Afshar\inst{1} \and
Amihood Amir\inst{2} \and
Michael T. Goodrich\inst{1}\orcidID{0000-0002-8943-191X} \and
Pedro Matias\inst{1}\orcidID{0000-0003-0664-9145}
}

\authorrunning{R.\,Afshar, A.\,Amir, M.\,T. Goodrich, and P.\,Matias}

\institute{Dept. of Computer Science, Univ. of California Irvine, USA
\email{\{afsharr,goodrich,pmatias\}@uci.edu} \and
Dept. of Computer Science, Bar Ilan Univ., Israel
\email{amir@cs.biu.ac.il}
}

\begin{document}
\maketitle

\begin{abstract}
We study the query
complexity of exactly reconstructing a string from
adaptive queries, such as substring, subsequence, and jumbled-index
queries.
Such problems have applications, e.g., in computational biology.
We provide a number of new and improved bounds for exact string reconstruction
for settings where either the string or the queries are ``mixed-up''.

For example, we show that a periodic (i.e., ``mixed-up'') 
string, $S=p^kp'$, of smallest period $p$, where $|p'|<|p|$,
can be reconstructed using $O(\sigma|p|+\lg n)$ substring queries, where
$\sigma$ is the alphabet size, if $n=|S|$ is unknown.
We also show that we can reconstruct $S$ after having been corrupted by a small number of errors $d$, measured by Hamming distance. In this case, we give an algorithm that uses $O(d\sigma|p| + d|p|\lg \frac{n}{d+1})$ queries. 
In addition, we show that a periodic string can be reconstructed 
using $2\sigma\lceil\lg n\rceil + 2|p|\lceil\lg \sigma\rceil$ subsequence queries,
and that general strings can be reconstructed using
$2\sigma\lceil\lg n\rceil + n\lceil\lg \sigma\rceil$ subsequence queries,
without knowledge of $n$ in advance.
This latter result improves the previous best, decades-old result,
by Skiena and Sundaram.
Finally, we believe we are the 
first to study the exact-learning query complexity for
string reconstruction using jumbled-index queries,
which are a ``mixed-up'' type of query that have received much attention of late.

\end{abstract}

%!TEX root = main.tex

\section{Introduction}
\emph{Exact learning} involves asking a series of queries so as to 
learn a configuration or concept uniquely and without errors, 
e.g., see~\cite{angluin1988queries}.
For example,
imagine a game where a player, Alice, is trying to exactly learn a secret string, 
$S$, 
such as $S=$~\texttt{"rumpelstiltskin"}, which is known only to a magic fairy.
Alice may ask the fairy questions about $S$, but only if they are in a form
allowed by the fairy,
such as ``Is $X$ a substring of $S$?''.
Any allowable question that Alice asks must be answered truthfully 
by the fairy.
Alice's goal is to learn $S$
by asking the fewest number of allowable questions.
Her strategy is \emph{adaptive} if her questions can depend on the answers to
previous queries.
This exact-learning string-reconstruction 
problem might at first seem like a contrived game,
but it actually has a number of applications.

For example, the magic fairy could represent a corporation
with a document database, $S$, that supports
an API allowing users to perform certain online 
query operations on $S$, such as keyword searches.
Further, this corporation may receive financial compensation for each of its
database responses 
(either directly or through advertisements); hence, the corporation
might not want the database's entire contents leaking out.
In this case, Alice could represent a rival corporation that is interested in 
learning the contents of the database, by asking legal queries from its API,
so that Alice can setup a competing online query service.
An optimal solution to the fairy-querying game would allow Alice 
to steal the database 
by asking the fewest number of questions necessary.

As another example,
in interactive DNA sequencing,
the fairy's string is an unknown DNA sequence, $S$,
and allowable queries are ``Is $X$ a substring of $S$?''
Each such question can be answered by a hybridization experiment that exposes 
copies of $S$ to a mixture containing specific primers to see which ones bind to $S$,
e.g., see~\cite{DBLP:journals/jcb/SkienaS95}.
An efficient scheme for Alice to play this fairy-querying game results in an
efficient method for sequencing the unknown DNA sequence.

Yet another application comes from computer security and
cryptography, dealing with
searchable encryption 
(e.g.,~\cite{DBLP:journals/jcs/CurtmolaGKO11,DBLP:conf/ndss/StefanovPS14}), 
where a database returns encrypted answers in
response to queries.
In this case, so long as Alice can, for instance, tell 
encryptions of ``yes'' apart from encryptions of ``no,''
then the fairy-querying game corresponds to a type 
of side-channel attack, e.g.,
see~\cite{DBLP:conf/ccs/KellarisKNO16,DBLP:conf/sp/LachariteMP18,DBLP:conf/ccs/NaveedKW15,DBLP:conf/ccs/CashGPR15,%
k-nn-attack,DBLP:conf/uss/ZhangKP16}.

Thus,
we are interested in the exact-learning complexity 
of adaptively learning an unknown string via queries of various given types, that is,
for exactly reconstructing a string from queries.
Formally, we are interested in minimizing a \emph{query-complexity}
measure, $Q(n)$, which, in our case, is the number of queries of certain types needed in
order to exactly learn a string, $S$.
This query-complexity concept
comes from machine-learning and complexity theory,
e.g., see~\cite{angluin1988queries,DBLP:conf/birthday/AfshaniADDLM13,CHOI2010551,%
Dobzinski:2012:QCC:2213977.2214076,Tardos1989,%
Yao:1994:DTC:195058.195414,BERNASCONI2001113}.

\subsection{Related Work}
Motivated by DNA sequencing, 
Skiena and Sundaram~\cite{DBLP:journals/jcb/SkienaS95}
were the first to study exact string reconstruction from adaptive queries.
For \emph{substring queries}, of the form ``Is $X$ a substring of $S$?'',
they give a bound for $Q(n)$ of $(\sigma-1)n+2\log n+O(\sigma)$, 
where $\sigma$ is the alphabet size.
For \emph{subsequence queries}, of the form ``Is $X$ a subsequence of $S$?'',
they prove a bound for $Q(n)$ of $\Theta(n\log \sigma + \sigma\log n)$.
Recently, Iwama {\it et al.}~\cite{iwama2018reconstructing} study the 
problem for binary alphabets, which 
removes the additive logarithmic term in this case.
These papers do not consider ``mixed-up'' strings, however,
such as strings that are periodic or periodic with errors.
% Ami addition
The abundance of repetitions and periodic runs in genomic sequences is 
well known and has been exploited in the last decades for biologic and
medical information (see e.g.~\cite{D-89,BW-94,b:99,KBK1:03,PFAP03,
WexlerYKG04, DomanicP07,SokolDB09,PellegriniRV10,dlle:16}). It is somewhat
surprising that this phenomenon has not been used to achieve more 
efficient algorithms.
% end ami addition
Margaritis and Skiena~\cite{ms-95} study a 
parallel version of exact string reconstruction
from queries, which are hybrids of adaptive and non-adaptive strategies, showing, e.g.,
that a length-$n$ string can be reconstructed in
$O(\log^2 n)$ rounds using $n$ substring queries per round.
Tsur~\cite{Tsur} gives a polynomial approximation 
algorithm for the 1-round case. 
As in~\cite{DBLP:journals/jcb/SkienaS95},
these papers do not consider bounds for $Q(n)$ based on properties of
the string such as its periodicity.
Cleve {\it et al.}~\cite{Cleve} study string reconstruction
in a quantum-computing model, showing, 
for example, that a sublinear number of queries
are sufficient for a binary alphabet.
This result does not seem to carry over 
to a classical computing model, however,
which is the subject of our paper.

Another type of query we consider 
is the
\emph{jumbled (or histogram)-index} query, first considered in \cite{DBLP:conf/stringology/CicaleseFL09,DBLP:journals/ijfcs/BurcsiCFL12,DBLP:journals/jcb/EresLP04,DBLP:journals/ipl/ButmanEL04} and studied more recently in, e.g. \cite{afshani-20,ami-jumbled-14,AMIR2016146,doi:10.1098/rsta.2013.0132,MOOSA2010795,Jumbled-13}. 
% Ami - added May 22
Jumbled indexing has many applications. It can be used as a tool for de novo peptide identification (as in 
e.g.~\cite{Kim09a,Kim09b,Jeong10}), and has been used as a filter for searching an image database
~\cite{toy:01,Cieplinski:01,DKN:08,WH:11,ZLT:17}. 
% end Ami addition
In this query, which has received much study of late,
but has not been studied before for adaptive string reconstruction,
one is given a Parikh vector, i.e., a vector of frequency counts for each
character in an alphabet, and asked if there is a substring of the reference string,
$S$, having these frequency counts and, if so, where it occurs in $S$. 
% Ami - added May 22
Such reconstruction may aid in narrowing down peptide identification, 
or focusing on image retrieval.
% and Ami addition

Another model for string reconstruction, tangential to ours and
studied extensively, is the one defined by a non-adaptive oracle,
where: we are given a set of answers to queries in
advance, and we aim to understand sufficient and necessary conditions
on the answers that enable the exact reconstruction of the
string. 
This model differs from the adaptive one considered
in this paper in that it focuses on the study of combinatorial
properties of strings, rather than on minimizing the number of
queries. Below, we give a detailed review of existing literature on this model, for each type of query considered in this paper.

\subsubsection{Non-adaptive Substring Queries}There is an extensive
line of work focusing on the ability to reconstruct a string given
the multiset of all its length-$L$ substrings. For $L\ge a\lg n$
($a > 1$), it is shown in
\cite{DBLP:conf/isit/GabrysM18,DBLP:conf/isit/ElishcoGMY19,DBLP:journals/tit/ChangCEK17}
that, as $n$ approaches infinity, almost every length-$n$ string
can be recovered. The following variants have also been studied:
(i) only a subset of the length-$L$ substrings is given, or each
substring is subject to substitution errors of fixed Hamming distance
\cite{DBLP:journals/corr/abs-1912-11108,DBLP:journals/tit/KiahPM16};
(ii) the hidden string is an i.i.d. DNA string
\cite{DBLP:journals/jcb/ArratiaMRW96}, combined with a random subset
of the length-$L$ substrings \cite{DBLP:journals/tit/MotahariBT13},
subject to probabilistic substitution errors
\cite{DBLP:conf/isit/MotahariRTM13} or edit errors (of fixed maximum
amount) \cite{DBLP:conf/isit/GangulyMR16}; (iii) the hidden string
satisfies several constraints based on its repeat statistics
\cite{bresler2013optimal,DBLP:journals/tcs/Ukkonen92} and input
substrings are subject to erasure errors\footnote{A letter in the
substring is replaced by an $\varepsilon$.}
\cite{DBLP:conf/isit/ShomoronyCT15}; and (iv) when partial
reconstruction of the hidden string is sufficient
\cite{DBLP:conf/isit/ShomoronyKXCT16}. On a different note, the
authors of \cite{DBLP:journals/tcs/FiciMRS06,DBLP:journals/tcs/CarpiL01}
consider instead the case where the input is a special set of
substrings which is derived from the set of maximal substrings.

\subsubsection{Non-adaptive Subsequence Queries}Perhaps the most
studied problem in this category is the $k$-deck problem: given the
multiset of all length-$k$ subsequences of a length-$n$ string $S$,
what is the smallest value of $k$ that enables the unique reconstruction
of $S$? This problem was introduced in \cite{kalashnik1973reconstruction},
who showed an upper bound of $\lfloor n/2 \rfloor$. This bound was
improved to $(1+o(1))\sqrt{(n \ln n)}$ in \cite{DBLP:journals/dm/Scott97}
and, in the same year, to $\lfloor 16/7 \sqrt{n} \rfloor + 5$ in
\cite{DBLP:journals/jct/KrasikovR97}. The first non-trivial lower
bound, of $\lg / \lg \lg n$, was given in \cite{zenkin1984non} and
later on, was improved to $\lg n$ in \cite{DBLP:journals/dm/ManvelMSSS91}
and to $e^{\Omega(\sqrt{\lg n})}$ in \cite{DBLP:journals/jct/DudikS03}.
Recently, Gabrys \textit{et al.} \cite{DBLP:conf/isit/GabrysM17}
considered an extension of the $k$-deck problem, where one is also
given a number of special subsequences of length $n-t$, $t>0$; they
provide lower and upper bounds that have a dependence on $t$. Also
related to the $k$-deck problem is the work of Simon
\cite{DBLP:conf/automata/Simon75}, on which subsequences are
considered to be of length \emph{at most} $k$. Another relevant
problem is trace reconstruction. The input to this problem is a set
of traces, distorted versions of the hidden string obtained by
deletion (i.e. subsequences) or other types of errors, when sending
it through a noisy channel. Similarly, the goal is to recover the
hidden string $S$, either exactly or with some accuracy or probability,
using the least amount of traces. To the best of our knowledge,
this problem was first studied in \cite{DBLP:journals/tit/Levenshtein01},
who provided bounds for the number of input traces, when subject
to a worse case fixed number of substitutions, transpositions,
deletions or insertion errors. In the case of exclusively dealing
with deletions, where each letter is deleted with some fixed
probability $q$, Batu \textit{et al.} \cite{DBLP:conf/soda/BatuKKM04}
showed that reconstruction is possible w.h.p. for $q=O(1/\lg n)$
and $O(\lg n)$ traces, when $S$ is chosen uniformly at random.
Moreover, they show that, for arbitrary $S$ and for
$q=O(1/n^{1/2+\epsilon})$, $O(1/\epsilon)$ traces are sufficient
to reconstruct a close approximation of $S$ and $O(n \lg n)$ traces
are sufficient to recover $S$ exactly. Later Kannan \textit{et al.}
\cite{DBLP:conf/isit/Kannan005} extended these results to the case
where insertion errors are also allowed, showing that for
deletion/insertion error probabilities of $q=O(1/\lg^2 n)$ and
$O(\lg n)$ traces, $S$ can be recovered w.h.p. assuming it is chosen
uniformly at random. Similarly, they show that an arbitrary $S$ can
be recovered w.h.p., for $q=O(1/n^{1/2+\epsilon})$ and $O(1)$ traces
of length at most $n^\epsilon$. Later, Viswanathan \textit{et al.}
\cite{DBLP:conf/soda/ViswanathanS08} improved on this, by showing
that deletion/insert error probabilities of $q=O(1/\lg n)$ are
sufficient to reconstruct $S$, chosen uniformly at random. They
also show that $\Omega(\lg n)$ traces are necessary to reconstruct
$1-o(1)$ length-$n$ strings w.h.p. In
\cite{DBLP:conf/soda/HolensteinMPW08}, the authors showed that, for
the case of deletion errors only, of probability $q=O(1)$,
reconstruction is possible w.h.p. using $poly(n)$ traces, when $S$
is chosen uniformly at random. Finally, Sala \textit{et al.}
\cite{DBLP:conf/isit/SalaGSMD16} studied lower bounds on the number
of input traces formed from a worst-case number of insertion errors,
where $S$ is a member of specific error-correcting codes, i.e. sets
of strings constructed strategically to allow recovering them from
a noisy channel is modified.

\subsubsection{Non-adaptive Jumbled-Index Queries} In
\cite{DBLP:conf/isit/AcharyaDMOP14,DBLP:journals/siamdm/AcharyaDMOP15},
Acharya \textit{et al.} study a non-adaptive version of the problem
of enumerating candidate strings from the \textit{composition multiset}
of the underlying string. The composition multiset corresponds to
the set of answers to all possible queries of the following type:
given a Parikh vector, how many times does a matching substring
occur in the hidden string? Under this model, they extend polynomial
techniques used for the turnpike problem (see
\cite{dakic2000turnpike,DBLP:conf/compgeom/SkienaSL90}) to give:
(i) sufficient (but not necessary) conditions for the ability to
uniquely reconstruct a string, (ii) a sufficient characterization
of unreconstructable strings and (iii) a backtracking algorithm
that enumerates the set of all candidate strings, whose cardinality
they lower and upper bound.

\subsection{Our Results}
We provide new and improved results for exactly reconstructing strings from 
adaptive substring, subsequence, and jumbled-index queries. 
For example, we believe we are the first to characterize
query complexities for exactly reconstructing periodic strings from adaptive
queries, including the following results
for reconstructing 
a length-$n$ periodic (i.e., ``mixed-up'')
string, $S=p^kp'$, of smallest period $p$, where $p'$ is a prefix of $p$ and the alphabet has size $\sigma$:
\begin{itemize}
\item
It requires at least $|p|\lg \sigma$ substring or subsequence queries.
\item
It can be done
with $\sigma|p|+\lceil\lg |p|\rceil$ substring queries, if $n$ is known.
\item
It can be done
with
$O(\sigma|p|+\lg n)$ substring queries, if $n$ is unknown.
\item
It can be done
with $\sigma\lceil\lg n\rceil + 2|p|\lceil\lg \sigma\rceil$
subsequence queries,
for known~$n$.
\item
It can be done
with $2\sigma\lceil\lg n\rceil + 2|p|\lceil\lg \sigma\rceil$
subsequence queries, if $n$ is unknown. 
\end{itemize}

Perhaps our most technical result is that we show 
that we can reconstruct a length-$n$ string, $S$, 
within Hamming distance $d$ of 
a periodic string $S'=p^kp'$, of smallest period $p$, using 
$O(\min(\sigma n, \, d\sigma|p| + d|p|\lg \frac{n}{d+1}))$ substring queries,
if $n$ is unknown. 
We also show
that we can exactly reconstruct a general length-$n$ string, $S$, 
using $2\sigma\lceil\lg n\rceil + n\lceil\lg \sigma\rceil$
subsequence queries, if $n$ is unknown. 
Such queries are another ``mixed-up'' setting, since there can be multiple
subsequence matches for a given string.
Our bound improves
the previous best, decades-old result, by 
Skiena and Sundaram~\cite{DBLP:journals/jcb/SkienaS95},
who prove a query complexity of $2\sigma\lg n+1.59n\lg \sigma+5\sigma$ for this
case. 
If $n$ is known, then $\sigma\lceil\lg n\rceil + n\lceil\lg \sigma\rceil$ subsequence queries suffice.
We believe we are the first to study string reconstruction using jumbled-index
queries, which are yet another ``mixed-up'' setting, since they simply count
the frequency of each character occurring in a substring.
We prove the following results:
\begin{itemize}
\item
We can reconstruct
a length-$n$ string
with $O(\sigma n)$ yes/no 
extended jumbled-index queries,
which include a count for an end-of-string character,~\$.
\item
For jumbled-index queries that return an index of a matching substring,
string reconstruction is not possible if this index is chosen 
adversarially, but is possible using $O(\sigma+n\lg n)$ queries if it
is chosen uniformly at random.
\end{itemize}

\subsection{Preliminaries}
\label{subsec:prelim}
We consider strings over the 
alphabet $\Sigma=\{a_1,a_2,\ldots,a_\sigma\}$ of $\sigma$ letters. The size of a string $X$ is denoted by $|X|$. 
We use $X[i]$ to denote the $i\textsuperscript{th}$ letter of $X$ and $X[i..j]$ to refer to the substring of $X$ starting at its $i\textsuperscript{th}$ and ending at its $j\textsuperscript{th}$ letter (e.g., $X=X[1..|X|]$). 
We may ignore $i$ when expressing a prefix $X[..j]$ of $X$. 
Similarly, $X[i..]$ is a suffix of $X$.
Occasionally, we will express concatenation of strings $X$ and $Y$ by $X\cdot Y$ (instead of $XY$) to emphasize some property of the string.
A string $X$ concatenated with itself $k$ (resp. infinitely many) times can be expressed as $X^k$ (resp. $X^\infty$). The reversal of a string $X$ is denoted by $X^R$.

A string, $S$, has \emph{period} $p$ if $S = p^kp'$, 
such that $k > 0$ is an integer and $p'$ is a (possibly empty) prefix of $p$. Further, a string $S$ is \emph{periodic} if it has a period that repeats at least twice, i.e. $S=p^kp'$ and $k>1$\footnote{Our algorithms assume that $S$ is periodic ($k>1$), while the Periodicity Lemma (\labelcref{lem:periodicity}) only requires a string to have a period ($k>0$).}. The following is a well known result concerning the periodicity of a string, due to Fine and Wilf~\cite{fine1965uniqueness}, which we will need later on.

\begin{lemma}[Periodicity Lemma \cite{fine1965uniqueness}]\label{lem:periodicity}
  If $p,q$ are periods of a string $X$ of length $|X|\ge |p|+|q|-\gcd(|p|,|q|)$, then $X$ also has a period of size $\gcd(|p|,|q|)$.
\end{lemma}

A \emph{doubling search} is the operation used
  to determine a number $n$ from a (typically unbounded) range of possibilities. It involves doubling a query value, $m$, until it
  is greater than $n$, followed by a binary search to determine $n$ itself. Its time complexity is $2 \lfloor\lg n\rfloor + 1$. A more sophisticated version of this procedure exists (see \cite{DBLP:journals/ipl/BentleyY76}) that actually improves the time complexity
  into $$\lfloor \lg L^{(0)} \rfloor + \lfloor \lg L^{(1)} n \rfloor + \dots + \lfloor \lg L^{(t-1)} n \rfloor + 2\lfloor \lg L^{(t)} n \rfloor + 1,$$ where $L^{(j)} (n)=\lfloor \lg L^{(j-1)} (n) \rfloor + 1$ and $L^{(0)} (n)=n$, for which there exists an optimized value of $t$.
For simplicity, we use the traditional algorithm, which is asymptotically equivalent.

\section{Substring Queries}
\label{sec:substr}

In this section,
we study query complexities for a string, $S$, 
subject to yes/no \emph{substring} queries, \textsf{IsSubstr}, i.e. queries of 
``Is $X$ a substring of $S$?''. We focus on the cases where $S$ corresponds to an originally periodic string, that may have lost its periodicity property due to error corruption. The nature of
the errors is context-dependent. For example, corruption may be
caused by transmission errors, measurement errors, malicious
tampering, or even by the aging process of a natural phenomenon.
There are multiple ways to model errors in strings.
Examples include:
\begin{itemize}
  \item \textit{Hamming} distance model: corruption is caused by allowing the substitution of a letter in the string by a different letter of the alphabet
  \item \textit{Edit} distance model: a generalization of the Hamming distance that also allows the insertion or deletion of a letter (see \cite{L-66})
  \item \textit{Swap} distance model: the operation allowed consists of swapping two adjacent letters in the string (see \cite{DBLP:journals/jacm/LowranceW75,DBLP:conf/stoc/Wagner75})
  \item \textit{Interchange} (or \textit{Caley}) distance model: it generalizes swap distance, by allowing the swap of any two letters, not necessarily adjacent (see \cite{cayley1849lxxvii,DBLP:journals/tcs/Jerrum85,DBLP:journals/jcss/AmirABLLPSV09,DBLP:conf/esa/AmirHKLP07})
\end{itemize}
In this paper, we consider Hamming distance.
We say that $S$ is a $d$-\emph{corrupted periodic string} if 
there exists a periodic string $S'$ of period $p$, such
that $|S|=|S'|$ and $\delta(S', S) \le d$, where $\delta$ is the
Hamming distance. We refer to $p$ as an \emph{approximate period} of $S$.
Notice that, depending on $d$, there might exist
multiple possible strings $S'$ that originate $S$.
We are interested
in reconstructing $S$, as opposed to $S'$, since we can use one of the existing algorithms
to enumerate all possible strings $S'$ (see
\cite{aelps:12,aals:18,krrswa:18}), without incurring additional
queries. 

Our main result in this section is the following.

\begin{restatable}{theorem}{thmSubstrIntercalated}
\label{thm:substr_error_intercalated}
  We can reconstruct a length-$n$ $d$-corrupted periodic string $S$ using $$O\left(\min \left(\sigma n, d\sigma|p| + d|p|\lg \frac{n}{d+1}\right)\right) \text{ queries,}$$for known $d$, unknown $|p|$, regardless of whether we know $n$, where $p$ is a smallest approximate period of $S$.
\end{restatable}

The algorithm of \cref{thm:substr_error_intercalated} is a more elaborate version of a reconstruction algorithm for the special case of $d=0$, i.e. when no errors occurred and $S=S'$, and when $n$ is not known in advance.

\begin{restatable}{theorem}{thmUnknown}
\label{thm:unknown}
We can reconstruct a length-$n$ periodic string, 
$S=p^kp'$, of smallest period $p$, using $O(\sigma |p| + \lg n)$ substring queries,
assuming both $n$ and $|p|$ are unknown in advance.
\end{restatable}

The algorithm of \cref{thm:unknown}, in turn, builds from a simple reconstruction algorithm that handles the case where $n$ is known in advance and $d=0$.

For clarity, we will present our results in increasing order of complexity, from the least general result of $d=0$ and known $n$, to the most general result of arbitrary $d$ and unknown $n$.

\subsection{Uncorrupted Periodic Strings of Known Size}
\label{subsec:uncorrupted:known}
We first give a simple algorithm to reconstruct a periodic string $S=p^kp'$ of smallest period $p$ and known size with query complexity $O(\sigma |p|)$, and then show how to improve this 
algorithm to have query complexity $\sigma|p|$ plus lower-order terms.
Our algorithms use a primitive 
developed by Skiena and Sundaram~\cite{DBLP:journals/jcb/SkienaS95},
which we call ``\emph{append} (resp., \emph{prepend}) a letter.''
In the append (resp., prepend) primitive,
we start with a known substring $q$ of $S$, and we ask 
queries \textsf{IsSubstr}$(qa_i)$
(resp., \textsf{IsSubstr}$(a_iq)$),
for each $a_i\in\Sigma$.
Note that if we know that one of the $qa_i$ (resp., $a_iq$) strings must
be a substring, we can save one query, so that appending or prepending a letter
uses at most $\sigma-1$ queries in this case.

In our simple algorithm,
we iteratively grow a candidate period, $q$, using the append primitive 
until $q^{g(q)-1}$ is a substring, 
where $g(x)=\lfloor n/|x| \rfloor$.
Notice that $q$ may be an ``unlucky'' cyclic rotation of $p$, 
which only repeats $g(p)-1$ times, and we need to account for this possibility. 
Thus,
once we get a substring corresponding to
$q^{g(q)-1}$, we then append/prepend letters until we recover all of $S$.
For reference, see \cref{alg:substr}, where the number of queries is shown in parentheses for steps involving queries.

\begin{algorithm}[t]
\caption{Reconstructing a periodic string $S=p^kp'$ of known size $n$ and smallest period $p$, for $k>1$.}
\label{alg:substr}
\DontPrintSemicolon
\SetNlSty{textbf}{}{.}
% \SetAlgoNLRelativeSize{-1}
\SetNlSkip{.8em}

\nl Let $q=\varepsilon$\;
\nl \Repeat{$\textsf{IsSubstr}(q^{g(q)-1})$ \complexity*[f]{1\mbox{~per iteration; $|p|$ iterations}}}{
  Append a letter to $q$ \complexity*{$\sigma - 1$}
  }
\nl Let $T=q^{g(q)-1}$\;
\nl While $T$ is a substring of $S$, append a letter to $T$\;%
\nl While $|T|<n$ and $T$ is a substring of $S$, prepend a letter to $T$
\hspace*{1em}\rlap{\smash{$\left.\begin{array}{@{}c@{}}\\{}\\{}\end{array}\color{blue}\right\}%
          \color{blue}\begin{tabular}{@{}l}$(\sigma(2|p|-1) )$\end{tabular}$}}\;
% above hack from: https://tex.stackexchange.com/questions/51019/how-can-i-put-a-curly-brace-inside-an-algorithm-to-group-code-lines 
\nl \textbf{Output} $T$
\end{algorithm}

\begin{restatable}{theorem}{warmupSubstring}
\label{lem:warmup-substring}
We can reconstruct a length-$n$ periodic string
$S=p^kp'$, of smallest period $p$, using $O(\sigma|p|)$ substring queries, assuming $n$ is known in advance and $|p|$ is unknown.
\end{restatable}

\begin{proof}
\label{proof:warmupSubstring}
  The main loop in \cref{alg:substr} will always terminate, because $S$ is periodic and any cyclic permutation of $p$ is a substring, when concatenated at least $g(p)-1$ times. It is easy to see that the procedure of iteratively appending letters to $q$ must result in a cyclic permutation of $p$, unless the main loop stops earlier. After the main loop, there are at most $2|p|-1$ letters left to be recovered, so the overall query complexity is at most $\sigma |p| + \sigma(2|p|-1)$, which is $O(\sigma|p|)$.
  \qed
\end{proof}

With a little more effort, we can improve the constant factor
in the query complexity.
The main challenge to achieving this improvement is
that, after the main loop in \cref{alg:substr}, $q$ may not
correspond to a cyclic rotation of $p$. For example,
in $S=abababaab\cdot abababaab\cdot abababaab$, we may get $q=abababa$, while the actual
period is $p=abababaab$. However, we show that, when $k=n/|p|>3$, the following implication holds indeed: if $q^{g(q)-1}$ is a substring, then $q$ must be a cyclic rotation of $p$.

We begin by giving the details 
for our improved algorithm for reconstructing
a periodic length-$n$ string $S$, when $n$ is known, 
which is shown in \cref{alg:substr_improved}.

\begin{algorithm}
\caption{Reconstructing a periodic string $S=p^kp'$ of known size $n$ and smallest period $p$, for $k>3$.}
\label{alg:substr_improved}
\DontPrintSemicolon
\SetNlSty{textbf}{}{.}
\SetNlSkip{.8em}
\SetAlgoVlined

\SetKwProg{myalg}{}{start}{end}
\myalg{}{
  \nl Let $q=\varepsilon$\;
  \nl \Repeat{$\textsf{IsSubstr}(q^{g(q)-1})$ \complexity*[f]{1\mbox{~per iteration; $|p|$ iterations}}}{
    \nl Append a letter to $q$ \complexity*{$\sigma -1$}
    }
  \nl Let $p=\textsf{TrueRotation}(q)$ \complexity*{$\lceil \lg |q|\rceil$}
  \nl Determine $p'$ and \textbf{output} $p^kp'$
}\;

\SetKwProg{myfunc}{function}{}{}
\SetKwFunction{truerotation}{TrueRotation}
\SetFuncSty{textsf}
\myfunc{\truerotation{$q$}}{
  Find, using binary search, the largest suffix $q[j..]$, such that \textsf{IsSubstr}$(q[j..]\cdot q^{g(q)-1})$ \complexity*{$\lceil \lg |q|\rceil$}
  \textbf{Return} $q[j..]\cdot q[..j-1]$
}
\end{algorithm}

\begin{remark}\label{rem:multiplicity}
  A string, $p$, is a period of a string $X$ of length $|X|\ge i|p|$ if and only if $p^j$ is a period of $X$, for all $j\in \{1,2,\dots,i\}$.
\end{remark}

\begin{restatable}{theorem}{theoremSubstrImproved}
\label{thm:known}
We can reconstruct a length-$n$ periodic string 
$S=p^kp'$, of smallest period $p$, using at most $\sigma |p| + \lceil\lg |p|\rceil $ substring queries, assuming that: $n$ is known in advance, $k>3$ and $|p|$ is unknown.
\end{restatable}

\begin{proof}
Consider \cref{alg:substr_improved}.
  We claim that, immediately after the main loop, 
the candidate period $q$ is indeed a cyclic rotation of the true period $p$. 
The remainder of the proof then follows from this.

So let us prove our claim.
Let $q$ be the string immediately after the main loop and let $T=q^{\lfloor n/ |q|\rfloor-1}$. If $|q|=|p|$, then $q$ is clearly a cyclic rotation of $p$. Besides, $|q|$ cannot be greater than $|p|$, because the letter-by-letter construction of $q$ would have implied a halt of the main loop when $q$ had size $|p|$: any cyclic rotation of $p$ must repeat at least $\lfloor n/|p| \rfloor -1$ times. So let us consider the case $|q| < |p|$. Since $k>3$, we have that $n \ge 4|p|$. Moreover, since $T=q^{\lfloor n/ |q|\rfloor-1}$, we know that $|T| \ge n - (2|q|-1)$ and, thus, $|T| \ge 2|p|$. Since $T$ is a substring of $S$, $T$ must have a second period of size $|p|$. Moreover,
  \begin{align*}
    |T| & \ge 2|p|\\
      & \ge |p|+|q|\\
      & \ge |p|+|q| - \gcd(|p|,|q|)
  \end{align*}
  Thus, by the Periodicity Lemma (\labelcref{lem:periodicity}), 
$T$ has a period $p_T$ of size $\gcd(|p|,|q|)$. Therefore, $S$ must have a period of size $|p_T|$, and thus, $S$ must have a period of size $|q|$ (by \cref{rem:multiplicity}), which contradicts the fact that $p$ is the smallest period of $S$.
\qed
\end{proof}

Our analysis above is tight in the sense that, for $k=3$, it no longer holds: recall the example given above, where $S=abababaab\cdot abababaab\cdot abababaab$ and $q=abababa$.

Notice that any reconstruction algorithm requires at least $|p|\lg \sigma$ queries; this follows from an information-theoretic argument.

\begin{restatable}{theorem}{theoremLowerBoundSubstr}
\label{thm:lower}
Reconstructing a length-$n$ string, $S=p^kp'$, of smallest period $p$, requires 
at least  $|p|\lg \sigma$ \textsf{IsSubstr} queries,
even if $n$ and $|p|$ are known.
\end{restatable}

\begin{proof}
There are $\sigma^{|p|}$ possible periods for $S$. 
Since each period corresponds to a different output of a reconstruction
algorithm, $A$, and each query is binary, we can model any such 
algorithm, $A$, as a binary decision tree,
where each internal node corresponds to an
\textsf{IsSubstr} query.
Each of the $\sigma^{|p|}$ possible periods must correspond to at least one
leaf of $A$; hence, the
minimum height of $A$ is $\lg (\sigma^{|p|})$.
\qed
\end{proof}

In the next section we consider the case where the underlying string is of unknown size.

\subsection{Uncorrupted Periodic Strings of Unknown Size}

As in \cref{subsec:uncorrupted:known}, we iteratively grow a candidate period $q$ and attempt to recover $S$ by concatenating $q$ with itself in the appropriate way. The difficulty when $n$ is unknown is that we can no longer confidently predict $g(q)$. Thus, we can no longer issue a single query to test if $q$ is the right period. An immediate solution is to use a doubling search. Unfortunately,
this introduces a multiplicative $O(\lg n)$ term into the query complexity. To avoid it, we show how we can take advantage of the Periodicity Lemma (\labelcref{lem:periodicity}) to amortize the extra work needed to recover $S$.

Let us describe the algorithm (see \cref{alg:substr_unknown_size} for reference). We start with an empty candidate period $q$. At each iteration, we add a letter to $q$, using the append primitive and, using a doubling search, determine the \emph{run-length} $t$ of $q$, i.e. the maximum integer $t$ such that $q^t$ is a substring of $S$. If $t=1$, we advance to the next iteration and repeat this process. If, on the other hand, $t>1$, we use $q$ to determine the largest substring $T$ that has a period of size $|q|$. This can be done efficiently, using doubling searches, by determining the largest suffix $l$ of $q$ and the largest prefix $r$ of $q$, such that $\textsf{IsSubstr}(l\cdot q^t\cdot r)$. Once $T$ is determined, we check whether it corresponds to $S$ by checking if there is any letter preceding and succeeding $T$ (see \textsf{IsValid} subroutine). If $T$ corresponds to $S$, we output it. Otherwise, we update $q$ to be any largest substring of $T$ whose size is assuredly less than $|p|$: using Periodicity Lemma (\labelcref{lem:periodicity}), we argue in \cref{lem:next_q} below that, if $q$ is not a cyclic rotation of $p$, then $p$ must be as large as \textit{almost} the entire substring $T$; more specifically, it must be the case that $|p|> |T|-|q|+1$. Thus, we update $q$ to be a length-$(|T|-|q|+1)$ prefix of $T$ (any other substring of $T$ would also work). We use this fact to get a faster convergence to a cyclic rotation of $p$, while making sure that we do not overshoot $|p|$. Indeed, this observation will enable us to incur a $O(\lg n)$ additive factor, instead of a multiplicative one. After updating $q$, we advance to the next iteration, where a new letter is appended to $q$, and repeat this process until $T=S$.

\begin{algorithm}[t]
\caption{Reconstructing a periodic string $S=p^kp'$, of smallest period $p$ and unknown size $n$, for $k>1$.}
\label{alg:substr_unknown_size}
\DontPrintSemicolon
\SetNlSty{textbf}{}{.}
\SetNlSkip{.8em}
\SetAlgoVlined

\SetKwProg{myalg}{}{start}{end}
\myalg{}{
\nl Let $q=\varepsilon$\;
\nl \Repeat{$\textsf{IsValid}(T)$ \complexity*[f]{$2\sigma$}}{
  \nl Append or prepend a letter to $q$ \complexity*{$\sigma-1$; potentially, $2\sigma-1$ when $k\le 2$} \label{alg:line:append_letter}
  \nl Determine the run-length $t$ of $q$ \complexity*{$2\lfloor \lg t\rfloor + 1$} \label{alg:line:runlength}
  \nl \lIf{$t=1$}{Let $T=q$}
  \nl \Else{
    \nl Let $l$ be the largest suffix of $q$ such that $\textsf{IsSubstr}(l\cdot q^t)$ \complexity*{$2\lfloor \lg |l|\rfloor + 1$} \label{alg:line:prefix}
    \nl Let $r$ be the largest prefix of $q$ such that $\textsf{IsSubstr}(l\cdot q^t\cdot r)$ \complexity*{$2\lfloor \lg |r|\rfloor + 1$} \label{alg:line:suffix}
    \nl Let $T=l\cdot q^t\cdot r$ \label{alg:line:T}\;
    \nl Let $q=T[..|T| - |q| + 1]$\label{alg:line:q_expanded} \;
    }
  }
\nl \textbf{Output} $T$
}\;

\SetKwProg{myfunc}{function}{}{}
\SetKwFunction{isvalid}{IsValid}
\SetFuncSty{textsf}
\myfunc{\isvalid{$T$}{\complexity*[f]{$2\sigma$}}}{
  Let $x$ be the letter to the left of $T$ or $\varepsilon$ if there is none \complexity*{$\sigma$}
  Let $y$ be the letter to the right of $T$ or $\varepsilon$ if there is none \complexity*{$\sigma$}
  \textbf{Return} $x==\varepsilon$ \textbf{and} $y==\varepsilon$\;
}
\end{algorithm}

\begin{lemma}\label{lem:next_q}
  Let $T$ be the largest proper substring of $S=p^kp'$, of smallest period $p$, such that: $|q|$ is the length of the smallest period of $T$. Then, $|p|> |T| - |q| + 1$.
\end{lemma}

\begin{proof}
  Let us assume, by contradiction, that $|p|\le |T| - |q| + 1$. Then, $|T|\ge |q| + |p| - 1$ and, thus, $|T|\ge |q| + |p| - \gcd(|q|,|p|)$. In addition, if $p$ is a period of $S$, then $T$ must have a period of size $|p|$. So, by the Periodicity Lemma (\labelcref{lem:periodicity}), $T$ also has a period of size $\gcd(|q|,|p|)$. Moreover, since $T$ is the largest proper substring of $S$, $|p|$ is not a multiple of $|q|$. Therefore, $T$ must have a period shorter than $|q|$, a contradiction.
\qed
\end{proof}

When $k\le 2$, our algorithm behaves similarly to the letter-by-letter algorithm of Skiena and Sundaram~\cite{DBLP:journals/jcb/SkienaS95} -- after finding a cyclic rotation $q$ of $p$, our algorithm will continue adding letters to $q$ until $q=S$, this time using both the append and prepend primitives.

Next, we give the details of the correctness and query complexity of \cref{alg:substr_unknown_size}. Let $q_1,q_2,\dots,q_m$ be the sequence of $m$ candidate periods of increasing length, each of which is the result of the append/prepend primitive at the beginning of every iteration (line~\ref{alg:line:append_letter} of \cref{alg:substr_unknown_size}), e.g. $|q_1|=1$.
Notice that each $q_i$ may be expanded (in line~\ref{alg:line:q_expanded}), so the difference $|q_i|-|q_{i-1}|$ may not necessarily be 1. In addition, let us use $t_i$ to denote the run-length of $q_i$ computed in line~\ref{alg:line:runlength}.

\begin{restatable}{lemma}{lemmaSuccessIf}
\label{lem:success_if}
  \cref{alg:substr_unknown_size} successfully returns $S=p^kp'$, of smallest period $p$, if there exists an iteration $i\in\{1,2,\dots,m\}$, such that $q_i$ is a cyclic rotation of $p$.
\end{restatable}

\begin{proof}
\label{proof:lemmaSuccessIf}
  If $t_i>1$, then it is easy to see that the string $T$, computed in line~\ref{alg:line:T} in iteration $i$, must correspond to $S$. If $t_i=1$, then the algorithm essentially switches to the letter-by-letter algorithm, appending or prepending letters until the end, when $q_m=S$. Correctness of the stopping condition follows from the correctness of \textsf{IsValid}.
\qed
\end{proof}

We now show that, indeed, at some iteration $i$, the candidate period $q_i$ is a cyclic rotation of $p$.

\begin{restatable}{lemma}{lemmaThereExistsi}
\label{lem:there_exists_i}
  There exists an iteration $i\in \{1,2,\dots,m\}$, such that $q_i$ is a cyclic rotation of $p$.
\end{restatable}

\begin{proof}
\label{proof:lemmaThereExistsi}
Let us assume that there is no such iteration $i$. Then, since all the $q_i$'s are increasing in length, it must be the case that there exists an iteration $j \in \{1,2,\dots,m-1\}$, such that: $|q_j|<|p|$, but $|q_{j+1}|>|p|$. However, it follows from \cref{lem:next_q} (when $t_j>1$) and the fact that we add a single letter to $q_j$ (when $t_j=1$) that $p$ must be at least as large as $q_{j+1}$, a contradiction.
\qed
\end{proof}

Let us now argue about query complexity.
The following lemma shows that we can charge the logarithmic factors, incurred in each iteration $j$, to the work that would have been required to find the letters introduced in $q_{j+1}$.
This establishes the amortization in query complexity. We denote the number of queries in iteration $j$ of \cref{alg:substr_unknown_size} by $\mathcal{Q}(j)$.

\begin{restatable}{lemma}{lemmaQueriesIteration}
\label{lem:queries_iteration}
  The number of queries $\mathcal{Q}(j)$ performed in iteration $j$ of \cref{alg:substr_unknown_size} is at most $\sigma(|q_{j+1}|-|q_j|)+O(\sigma)$, for $j<m$, or $O(\sigma + \lg n)$, for $j=m$.
\end{restatable}

\begin{proof}
\label{proof:lemmaQueriesIteration}
  Let $l_j$ and $r_j$ denote, respectively, the lengths of the prefix $l$ and suffix $r$ computed in lines~\ref{alg:line:prefix} and \ref{alg:line:suffix} of \cref{alg:substr_unknown_size} in iteration $j$. The query complexity in any iteration $j$ is
\[
\mathcal{Q}(j)\le 2\lfloor \lg t_j\rfloor + 1 + 2\lfloor \lg l_j\rfloor + 1 + 2\lfloor \lg r_j\rfloor + 1 + 4\sigma
\]
  Let us assume that $t_j>1$, since otherwise the query complexity is $O(\sigma)$ and, therefore, agrees with the query complexity that is stated in the lemma.

  When $j=m$, it must be the case that $q_m$ is a cyclic rotation of $p$, and therefore has size $|p|$. Thus, we spend at most: (i) $\sigma$ queries when appending the $p\textsuperscript{th}$ letter, (ii) $2\lfloor \lg n/|p| \rfloor+1$ queries to determine the run-length $t_m$, and (iii) $2(2\lfloor \lg |p|\rfloor + 1)$ queries to determine the suffix and prefix of lengths $l_m$ and $r_m$, respectively. Notice that the combined $\log$ factors result in no less than $\Theta(\lg n)$. Thus, when $j=m$, the overall query complexity is $O(\sigma + \lg n)$.

  When $j<m$, we have the following:
\begin{align*}
  \lg t_j & \le t_j - 1 && (t_j > 1)\\
  \implies \lg t_j & \le (t_j-2)q_j + 1 && (q_j \ge 1)\\
  \implies \lg t_j + \lg l_j + \lg r_j & \le (t_j-2)q_j + l + r + 1 && (\lg x < x)\\
  \implies 2(\lg t_j + \lg l_j + \lg r_j) & \le 2((t_j-2)q_j + l + r + 1)\\
  \implies \mathcal{Q}(j) & \le 2((t_j-2)q_j + l + r) + 2 + 3 + O(\sigma)&& (\text{def. of }\mathcal{Q}(j))\\
  \implies \mathcal{Q}(j) & \le 2((t_j-2)q_j + l + r + 2) + O(\sigma)\\
  \implies \mathcal{Q}(j) & \le 2(|q_{j+1}|-|q_j|) + O(\sigma) && (\star)\\
  \implies \mathcal{Q}(j) & \le \sigma(|q_{j+1}|-|q_j|) + O(\sigma) && (\sigma \ge 2),\\
\end{align*}
  where $(\star)$ follows from the fact that, when $t_j>1$, $|q_{j+1}|=(t_j-1)|q_j| + l + r + 2$.
\qed
\end{proof}

Finally, we are in conditions of proving \cref{thm:unknown}, recalled below for convenience:

\thmUnknown*
\begin{proof}
  Correctness follows from \cref{lem:success_if,lem:there_exists_i}. As for the query complexity, it follows from \cref{lem:queries_iteration}, that the overall query complexity of \cref{alg:substr_unknown_size} is
  \[
  \sum_{j=1}^m \mathcal{Q}(j)
  \]
  Let $i$ be the iteration in which $|q_i|=|p|$ (see \cref{lem:there_exists_i}) and let us consider the queries done up to and after iteration $i-1$. Thus, by \cref{lem:queries_iteration}:
  \begin{align*}
    \sum_{j=1}^m \mathcal{Q}(j) &\ = \sum_{j=1}^{i-1} \Big(\sigma(|q_{j+1}|-|q_j|) + O(\sigma)\Big) + \sum_{j=i}^{m} \mathcal{Q}(j)\\
                                &\ = O(\sigma|p|) + \sum_{j=i}^{m} \mathcal{Q}(j),
  \end{align*}
  where the last equality follows from the telescoping nature of the first summation. As for the second summation, regarding the queries done after iteration $i-1$, we consider two cases. If $i=m$, then we spend either $O(\sigma)$ queries if $t_i=1$, or $O(\sigma + \lg n)$ queries if $t_i>1$, by \cref{lem:queries_iteration}. If, on the other hand, $i<m$, then notice that it must have been the case that $t_j=1$ for all $j\in \{i,i+1,\dots,m\}$. Thus, the total number of letters in $S$ left to recover at the end of iteration $i-1$ is at most $2|q_i|-1=2|p|-1$, each of which is added during each iteration $j\in\{i,i+1,\dots,m\}$ using $O(\sigma|p|)$ queries in total. Thus, whether or not $i=m$, the overall query complexity is
  \[\sum_{j=1}^m \mathcal{Q}(j) = O(\sigma|p| + \lg n)\]
\qed
\end{proof}

\subsection{Corrupted Periodic Strings}
\label{subsec:corrupted}

Let us assume throughout the remainder of this section that
$S$ is a $d$-corrupted periodic string of approximate period $p$. Recall that $S$ is a $d$-corrupted periodic string if 
there exists a periodic string $S'$ of period $p$, such
that $|S|=|S'|$ and $\delta(S', S) \le d$, where $\delta$ is the
Hamming distance. Again, the main
idea of the algorithm described in this section consists of:
(1) determining a cyclic rotation of a true period (in this case,
there might be multiple true periods), by iteratively growing a
candidate period $q$, and (2) using $q$ to recover $S$ accordingly.
However, in the presence of errors, each of these steps becomes
more difficult to realize efficiently. For example, in the first
step, we might be growing a candidate period $q$ that includes an
error. So, in order to rightfully reject the hypothesis that $q$
is at most as large as some approximate period $p$, our algorithm should
be able to tell the difference between (i) $|p|=|q|$ and $q$ includes an error and (ii) $|p|>|q|$. Otherwise, the algorithm will keep on growing $q$ until
it is equal to $S$, possibly incurring $\sigma n$ queries. In
addition, the second step of using $q$ to determine $S$ requires
more work, since the presence of errors discards the possibility
of simply concatenating $q$ with itself the required number of
times.
Because of these issues, it is crucial that our
algorithm understands when a candidate period is or not
free of errors. Thus, the algorithm relies on the following.

\begin{lemma}\label{lem:must_be_true_rotation}
  Let $A$ be any length-$(2d+1)|p|$ substring of a $d$-corrupted periodic string $S$ of approximate period $p$, corresponding to the concatenation of length-$|p|$ substrings $q_1,q_2,\dots,q_{2d+1}$. Then, a cyclic rotation of $p$ must be the only substring $q_j$ appearing at least $d+1$ times in $q_1,q_2,\dots,q_{2d+1}$.
\end{lemma}

\begin{proof}
  Clearly, there is some $q_i$ that is a cyclic rotation of $p$. Moreover, there is some $q_j$ that appears at least $d+1$ times in $q_1,q_2,\dots,q_{2d+1}$, or the number of errors would exceed $d$, by the pigeonhole principle. If $i\neq j$, then each occurrence of $q_j$, contributes at least 1 error, resulting in at least $d+1$ errors, a contradiction. Finally, $q_j$ must be the only string with $d+1$ appearances in $q_1,q_2,\dots,q_{2d+1}$, by the pigeonhole principle.
\qed
\end{proof}

\begin{algorithm}
\caption{Reconstructing a $d$-corrupted periodic string $S$.}
\label{alg:substr_errors}
\DontPrintSemicolon
\SetNlSty{textbf}{}{.}
\SetNlSkip{.8em}
\SetAlgoVlined
\SetArgSty{text}

\nl Let $A=\varepsilon$\;
\nl \Repeat{\textsf{success}}{
\nl Append/prepend $\min(2d+1,|S|-|A|)$ letters to $A$ \complexity*{$\sigma(2d+2)$}
\nl Let $q$ be the candidate period that is a substring of $A$,\;
  \hspace{.5em} as determined by \cref{lem:must_be_true_rotation} \;
\nl $(\textsf{success}, T)=\textsf{Expand}(q)$ \complexity*{$O(d\sigma + d\lg \frac{n}{d+1})$}
  }
\nl \textbf{Output} $T$
\end{algorithm}

\begin{function}
\SetAlCapNameSty{textsf}
\caption{Expand($q$)\hfill\textcolor{blue}{$(O(d\sigma + d\lg \frac{n}{d+1}))$}}\label{alg:expand}
\DontPrintSemicolon
\SetNlSty{textbf}{}{.}
\SetNlSkip{.8em}
\SetAlgoVlined

\nl Let $T = q$, $\textsf{done}=\textsf{False}$\;
\nl \While{$\delta(T, q^\infty[..|T|]) \le d$ \textbf{and} \textsf{not done}}{
  \nl Find the largest substring $R$, such that $\textsf{IsSubstr}(T\cdot R)$ \complexity*{$2\lfloor\lg |R|\rfloor + 1$}\label{alg:line:R}
  \nl Find the largest substring $L$, such that $\textsf{IsSubstr}(L\cdot T\cdot R)$ \complexity*{$2\lfloor\lg |L|\rfloor + 1$}\label{alg:line:L}
  \nl Let $r$ $(l)$ be the letter to the right (left) of $L\cdot T \cdot R$ or $\varepsilon$ if there is none \complexity*{$2\sigma$} \label{alg:line:left_right}
  \nl Let \textsf{done} $=(r==\varepsilon$ \textbf{and} $l==\varepsilon)$\;
  \nl Let $T=l\cdot L\cdot T \cdot R\cdot r$\;
}
\nl \lIf{$\delta(T, q^\infty[..|T|]) > d$}{
  \Return $(\textsf{False}, \_)$
}
\nl \Return $(\textsf{True}, T)$
\end{function}

Let us give the details for our algorithm, which is able to recover $S$, even when its size $n$ is unknown (see \cref{alg:substr_errors} for reference). We maintain an initially empty substring, $A$, of $S$, by extending it with $2d+1$ letters in each iteration, using the append and prepend primitives (as described in \cref{subsec:uncorrupted:known}), potentially incurring an extra $\sigma$ queries for detecting a left or right endpoint of $S$. In the case that $n=|S|<|p|(2d+1)$, the last iteration requires only $\min(2d+1,|S|-|A|)$ new letters. Thus, after adding letters to $A$ in the $i\textsuperscript{th}$ iteration, $A$ is a substring of $S$ of size at most $i(2d+1)$. Before advancing to the next iteration, we determine the only possible length-$i$ candidate period $q$ that could have originated $A$ with at most $d$ errors (by \cref{lem:must_be_true_rotation}). At this point we do not know if some approximate period $p$ has size $|p|=i$, so we try to use $q$ to recover the rest of the string, halting whenever the total number of errors exceeds $d$, in which case we advance to the next iteration and repeat this process for a new candidate period of size $i+1$. This logic is in the subroutine \hyperref[alg:expand]{$\textsf{Expand}(q)$}, described next (see the pseudo-code for reference). It initializes a string $T$ to $q$ and expands it by doing the following at each iteration:
\begin{enumerate}
  \item Appending to $T$ the largest periodic substring of period $\overrightarrow{q}$, where $\overrightarrow{q}$ is the appropriate cyclic rotation of $q$ that aligns with the right-endpoint of $T$. This can be done efficiently by determining the maximum value of $x$, using a doubling search, for which $$\textsf{IsSubstr}(T\cdot (\overrightarrow{q}^\infty[..\ x])),$$ incurring $2\lfloor \lg x\rfloor+1$ queries. The cyclic rotation $\overrightarrow{q}$ can be determined with no additional queries, by maintaining the value $x'$, which is the value of $x$ in the previous iteration, i.e. $\overrightarrow{q}$ is the cyclic rotation of $q$ starting at the index $(x' \mod |q|+2)$ of $q$.
  
  \item Prepending to $T$ the largest periodic substring of period $\overleftarrow{q}$, where $\overleftarrow{q}$ is the appropriate cyclic rotation of $q$ that aligns with the left-endpoint of $T$. This can be done efficiently by determining the maximum value of $y$, using a doubling search, for which $$\textsf{IsSubstr}(((\overleftarrow{q}^R)^\infty[..\ y])^R\cdot T),$$ incurring $2\lfloor \lg y\rfloor+1$ queries. The cyclic rotation $\overleftarrow{q}$ can be determined with no additional queries in a similar fashion to $\overrightarrow{q}$.

  \item Determining, if they exist, the letters immediately to the left and to the right of $T$, using $2\sigma$ queries, and adding them to $T$.
\end{enumerate}

The expansion process in $\textsf{Expand}(q)$ halts when either the total number of errors with respect to $q$, $\delta(T, q^\infty[..|T|])$, exceeds $d$ (in which case we advance to the next iteration), or when $T=S$ (in which case we return $T$).
% Need this comment to save vertical space
\begin{remark}\label{rem:expand}
  $\textsf{Expand}(q)$ successfully returns $S$ if and only if $q$ is a cyclic rotation of some approximate period.
\end{remark}
% Need this comment to save vertical space
\begin{restatable}{lemma}{lemmaExpandComplexity}
\label{lem:expand_complexity}
  The number of queries performed during any call to \textsf{Expand} is $O(d\sigma + d\lg \frac{n}{d+1})$.
\end{restatable}

\begin{proof}
\label{proof:lemmaExpandComplexity}
  Each call to \textsf{Expand} uses at most $2(d+1)\sigma$ queries to determine the corrupted letters, as well as the left/right endpoints of $S$ -- the total number of iterations of the while loop in \textsf{Expand} is $d+1$, since every iteration except the last introduces at least 2 errors in $T$, and each iteration incurs $2\sigma$ queries.

  In addition, the number of queries used by $\textsf{Expand}(q)$ during the doubling searches is 
\[
\sum_{j=1}^{|q|} \left(2\lfloor\lg R_j\rfloor + 2\lfloor\lg L_j\rfloor + 2\right),
\] 
where 
$R_j$ and $L_j$ denote, respectively, the lengths of the substrings determined via doubling searches in lines~\ref{alg:line:R} and \ref{alg:line:L}, during the $j\textsuperscript{th}$ call to \textsf{Expand}. Since the total number of iterations is $d+1$, there is at most $d+2$ such $R_j$'s and $L_j$'s. Moreover, the above summation is maximized when all the $R_j$'s and $L_j$'s have the same average value of at most $(n-d)/(d+1)$. This follows from Jensen's inequality and concavity of $\log$. Thus, the overall time complexity is $$O\left(d\sigma + d\lg \frac{n}{d+1}\right).$$
\qed
\end{proof}

Correctness and query complexity of our algorithm follows from \cref{rem:expand,lem:must_be_true_rotation,lem:expand_complexity}, giving us the following result:

\begin{restatable}{theorem}{thmSubstrError}
\label{thm:substr_error}
  We can reconstruct a length-$n$ $d$-corrupted periodic string $S$ using $O(d\sigma|p| + d|p|\lg \frac{n}{d+1})$ queries, for known $d$, unknown $|p|$, regardless of whether we know $n$, where $p$ is a smallest approximate period of $S$.
\end{restatable}

\begin{proof}
\label{proof:thmSubstrError}
  At the $|p|^\textsuperscript{th}$ iteration of the main loop, $A$ has size $(2d+1)|p|$ and, by \cref{lem:must_be_true_rotation}, $q$ must correspond to a cyclic rotation of some approximate period $p$. Correctness of reconstruction then follows from \cref{rem:expand}.

  The overall query complexity consists of the queries used to expand $A$ in each iteration and the queries used in the calls to the subroutine \textsf{Expand}. The former requires at most $(2d+2)\sigma|p|$ queries overall, and the latter requires at most $O(d\sigma|p| + d|p|\lg \frac{n}{d+1})$, by \cref{lem:expand_complexity}. Thus, the overall query complexity is $$O\left(d\sigma|p| + d|p|\lg \frac{n}{d+1}\right).$$
\qed
\end{proof}

If $n$ is known, we could save the queries used to check the left and right endpoints of $S$ in line~\ref{alg:line:left_right} of \textsf{Expand}, but this does not alter the query complexity asymptotically.

We assume a small enough number of errors, following~\cite{aelps:12}. \Cref{alg:substr_errors} is an improvement to the $O(\sigma n)$ letter-by-letter algorithm of Skiena and Sundaram~\cite{DBLP:journals/jcb/SkienaS95} for general strings of known and unknown size, when $d=O(\sigma n / (\sigma |p| + |p|\lg n))$. In particular, if $d = O(k/(1+\lg n))$, then \cref{alg:substr_errors} is an improvement, where $k=\lfloor n/|p|\rfloor$. Thus, our algorithm performs better if there is, on average, at most 1 error in every other $O(1 + \lg n)\textsuperscript{th}$ non-overlapping occurrence of $p$. If the number of errors is not small enough, then one should run the letter-by-letter algorithm intercalated with ours, to get an upper bound of $O(\sigma n)$ queries, giving us \cref{thm:substr_error_intercalated}, which we referred to at the beginning of \cref{sec:substr}, recalled here for convenience.

\thmSubstrIntercalated*

\section{Subsequence Queries}
\label{sec:subsequence}

%!TEX root = main.tex

We study the query complexity for a length-$n$ string, $S$,
subject to yes/no \emph{subsequence} queries, 
\textsf{IsSubseq}, i.e., queries of the 
form ``Is $X$ a subsequence of $S$?''

We begin with a simple lower bound.

\begin{restatable}{theorem}{thmLowerBoundSubsequence}
\label{thm:lower2}
Reconstructing
a length-$n$ periodic string, $S=p^kp'$, of smallest period $p$, requires 
at least $|p|\lg \sigma$ \textsf{IsSubseq} queries,
even if $n$ and $|p|$ are known.
\end{restatable}

\begin{proof}
\label{proof:lowerbound:subsequence}
The proof follows that of \cref{thm:lower} for substring queries, which can be found in \cref{sec:substr}.
\qed
\end{proof}

Let us next describe an algorithm for reconstructing a periodic
length-$n$ periodic string,
$S=p^kp'$, of smallest period $p$.
We begin by performing either binary searches (if $n$ is known) or doubling search (if $n$ is unknown), using queries of the form
\textsf{IsSubseq}($a^i$) to determine the number of $a$'s in $S$,
for each $a\in\Sigma$.
From all of these queries, we can determine the value of $n$ if
it was previously unknown.
This part of our algorithm requires either
$\sigma\lceil\lg n\rceil$  or
$2\sigma\lceil\lg n\rceil$ queries in total, depending on whether we knew $n$
at the outset.

If the number of $a$'s in $S$ is $n$, for any $a\in\Sigma$,
then we are done, so let us assume the number of $a$'s in $S$ 
is less than $n$, for each $a\in\Sigma$.
Thus, when we complete all our doubling/binary searches,
for each letter, $a \in\Sigma$ that occurs a nonzero number of times
in $S$,
we have a maximal subsequence, $S_a$, of $S$,
consisting of $a$'s.  Moreover, since $S$ is periodic with a
period that repeats $k$ times, each $S_a$ is periodic with a period that
repeats $k$ times. Unfortunately, at this point in the algorithm, we may not
be able to determine $k$.
So next we create a binary merge tree, $T$, with each of its leaves
associated with a nonempty subsequence, $S_a$, much in the style
of the well-known merge-sort algorithm, so that $T$ has
height $\lceil \lg \sigma\rceil$.
We then perform a bottom-up merge-like procedure in $T$ using
\textsf{IsSubseq} queries, as follows.

Let $v$ be an internal node in $T$, with children $x$ and $y$ for which we
have inductively determined periodic subsequences, $S_x$ and $S_y$,
respectively, of $S$.
Let $n_x=|S_x|$ and $n_y=|S_y|$. To create the subsequence, $S_v$, for $v$,
we need to perform a merge procedure to interleave $S_x$ and $S_y$.
To do this, we maintain indices $i$ and $j$ in $S_x$ and $S_y$, respectively,
such that we have already determined an interleaving,
$S_v[..i+j]$, of $S_x[..i]$ and $S_y[..j]$.
Initially, $i=j=0$.
We then perform the query
\textsf{IsSubseq}($S_v[..i+j]\cdot S_x[i+1]\cdot S_y[j+1..n_y]$).
Suppose the answer to this query is ``yes''.
In this case, we set
$S_v[..i+j+1]= S_v[..i+j]\cdot S_x[i+1]$
and we increment $i$.
If, on the other hand, the answer to the above query is ``no'', then
we set $S_v[..i+j+1]= S_v[..i+j]\cdot S_y[j+1]$,
because in this case we know that 
\textsf{IsSubseq}($S_v[..i+j]\cdot S_y[j+1]\cdot S_x[i+1..n_x]$) would
return ``yes''.
If this latter condition occurs, then we increment $j$.

Let ${q_v}$ denote this new interleaving prefix, $S_v[..i+j]$,
and let ${\hat k}=\lfloor n/|{q_v}|\rfloor$.
If ${q_v}^{\hat k}{q_v}'$ is a plausible interleaving of 
$S_x$ and $S_y$, where ${q_v}'$ is a prefix of ${q_v}$,
then we next ask the query
\textsf{IsSubseq}(${q_v}^{\hat k}{q_v}'$).
If the answer is ``yes'', then we set
$S_v={q_v}^{\hat k}{q_v}'$ and this completes the merge.
Otherwise, we continue incrementally interleaving $S_x$ and $S_y$,
using the current values of $i$ and $j$, by iterating
the procedure described above.
Clearly, this merge procedure asks at most $2|{q_v}|$ queries in total.

\begin{lemma}\label{lem:subseq_qv}
Let $p_v$ be the subsequence of $p$ consisting of the letters
from $S_v$.
Then $|{q_v}|\le |p_v|$.
\end{lemma}

\begin{proof}
The letter-by-letter construction of $q_v$ ensures that $q_v$ is the smallest period of $S_v$. Since $p_v$ is itself a period of $S_v$ no smaller than the smallest, then $|p_v|\ge |q_v|$.
\qed
\end{proof}

\begin{restatable}{theorem}{thmSubsequence}
\label{thm:subseq}
We can determine a length-$n$ periodic string, $S=p^kp'$, of smallest period $p$ of unknown size, using
$2\sigma\lceil \lg n\rceil + 2|p|\lceil\lg \sigma\rceil$ \textsf{IsSubseq} queries,
if $n$ is unknown.
If $n$ is known, then
$\sigma\lceil \lg n\rceil + 2|p|\lceil\lg \sigma\rceil$ 
\textsf{IsSubseq} queries suffice.
\end{restatable}

\begin{proof}
The total query complexity includes: (i) the letter decomposition $S_a$ for all $a\in\Sigma$, during the the first stage and (ii) the merge-like composition of all subsequences $S_a$, during the second stage.
If $n$ is known, the first stage requires $|\Sigma|$ binary searches, incurring $\sigma\lceil \lg n\rceil$ queries. Otherwise, it requires $|\Sigma|$ doubling searches, amounting to $2\sigma\lceil \lg n\rceil$ queries.
Regarding the second stage, we claim that any level $l$ of the binary merge tree, $T$, incurs a total of at most $2|p|$ queries, which amounts to a total of at most $2|p|\lceil\lg \sigma\rceil$ queries, when taking into account all the $\lceil \lg \sigma \rceil$ levels of $T$. Let $T(l)$ be the set of all nodes in $T$ at level $l$. Then, $$\sum_{v\in T(l)} |q_v| \le \sum_{v\in T(l)} |p_v| = |p|.$$ This follows from \cref{lem:subseq_qv} and the fact that all $\{S_v\ |\ v\in T(l)\}$ are pairwise letter-disjoint. Since the merge of an internal node $v$ requires a cost of $2|q_v|$, the total cost incurred in any level $l$ of $T$ is at most $2|p|$.
\qed
\end{proof}

A simple modification
of our algorithm also implies the following.

\begin{restatable}{theorem}{thmSubsequenceGeneral}
We can determine a length-$n$ string, $S$, using
$2\sigma \lceil\lg n\rceil + n\lceil\lg \sigma\rceil$ \textsf{IsSubseq} queries,
without knowing the value of $n$ in advance.
If $n$ is known, then
$\sigma\lceil \lg n\rceil + n\lceil\lg \sigma\rceil$ 
\textsf{IsSubseq} queries suffice.
\end{restatable}

\begin{proof}
\label{proof:thmSubsequenceGeneral}
Modify our subsequence-querying algorithm given in \cref{sec:subsequence} to remove
the queries for strings of the form ${q_v}^{\hat k}{q_v}'$.
The proof follows by
an analysis similar to that for \cref{thm:subseq}.
\qed
\end{proof}

This latter theorem improves a result of 
Skiena and Sundaram~\cite{DBLP:journals/jcb/SkienaS95},
who prove a query bound
of $2\sigma \lg n + 1.59n\lg \sigma+5\sigma$ 
when $n$ is unknown.

\section{Jumbled-index Queries}

Jumbled-indexing involves preprocessing a given string, $S$, 
so as to determine whether there exists a substring of $S$ whose
letter frequencies match the given \emph{Parikh vector}, i.e., a vector
$\psi = (f_1, \dots, f_\sigma)$ such that $f_i$ is the number
of occurrences in $S$ of $a_i\in\Sigma$, 
e.g., see~\cite{afshani-20,ami-jumbled-14,AMIR2016146,doi:10.1098/rsta.2013.0132,MOOSA2010795,Jumbled-13}.
In this section, we study the query complexity for
reconstructing an unknown length-$n$ string, $S$, using jumbled-index
queries. As observed by Acharya \textit{et al.} \cite{DBLP:conf/isit/AcharyaDMOP14,DBLP:journals/siamdm/AcharyaDMOP15}, strings and their reversals have the same ``composition multiset''. This immediately implies the following negative result, which we prove regardless for completeness.

\begin{lemma}
If $S$ is not a palindrome,
then $S$ cannot be reconstructed by yes/no jumbled-index queries, which
return whether there is a substring in $S$ with a given Parikh vector.
\end{lemma}

\begin{proof}
Suppose $S\not=S^R$, where $S^R$ denotes 
the reversal of $S$.
For any substring, $T$, of $S$, there is, of course, a corresponding substring,
$T^R$, of $S^R$.
Moreover, $T$ and $T^R$ have the same Parikh vector.
Thus, $S$ and $S^R$ have the same set of responses to yes/no 
jumbled-index queries; hence, any set of yes/no jumbled-index queries
cannot distinguish $S$ from $S^R$.
\qed
\end{proof}

Given that simple yes/no jumbled-index queries are not sufficient 
for string reconstruction, let us consider an extended type of 
yes/no jumbled-index query.

\begin{itemize}
  \item \emph{Jumbled-Indexing with End-of-string symbol ``\$''} (JIE): 
given an \emph{extended} Parikh vector,
$\psi = (f_1, \dots, f_\sigma,f_\$)$, for the letters in $\Sigma$ and
an end-of-string symbol, \$, which is not in $\Sigma$,
this query returns a yes/no response as to whether there is a substring of $S\$$
with extended Parikh vector $\psi$.
\end{itemize}

\noindent
Unlike the yes/no jumbled-index queries,
this variant enables full reconstruction.

\begin{theorem}
We can reconstruct a length-$n$ string, $S$, using $(\sigma-1)n$ JIE queries,
if $n$ is known, or $\sigma(n+1)$ JIE queries, if $n$ is unknown.
\end{theorem}

\begin{proof}
Our method is to use
a letter-by-letter reconstruction algorithm via an adaption
of the prepend-a-letter primitive for substring queries.
Suppose $n$ is unknown. 
Let $\psi$ be an extended Parikh vector for a known suffix, $s$, of $S$\$; initially,
$\psi=(0,0,\ldots,0,1)$ and $s=\$$.
Then we perform a jumbled-index query for $\psi_i$, for each $a_i\in\Sigma$,
where $\psi_i=\psi$ except that $\psi_i$ adds $1$ to the $f_i$ value in $\psi$.
If one of these, say, $\psi_i$, returns ``yes'', then we prepend $a_i$ to our known
suffix and we repeat this procedure using $\psi_i$ for $\psi$.
If all of these queries return ``no'', then we are done.
If $n$ is known, on the other hand,
then we can skip this last test of all-no responses and we can
also save at least one query with each iteration, with the algorithm otherwise
being the same.
\qed
\end{proof}

We can also consider jumbled-index queries that return an index of a matching
substring for a given Parikh vector, if such a substring exists. Though related, notice that this type of query is not subsumed by the query studied in Acharya \textit{et al.} \cite{DBLP:conf/isit/AcharyaDMOP14,DBLP:journals/siamdm/AcharyaDMOP15}, which returns the number of occurrences (instead of position) of matching substrings in $S$. There is some ambiguity, however, if there is more than one matching substring;
hence, we should consider how to handle such multiple matches.
For example, if a jumbled-index query returns the indices of all matching
substrings, then $\sigma$ queries are clearly sufficient to reconstruct any
length-$n$ string, for any $n$, without knowing the value of $n$ in advance.
Thus, let us consider two more-interesting types of jumbled-index queries.

\begin{itemize}
  \item \sloppy  % to allow line breaking
  \emph{Adversarial Jumbled-Indexing} (AJI): 
given a Parikh vector,
$\psi = (f_1, \dots, f_\sigma)$, this query
returns, in an adversarial manner, one of the starting indices of 
a matching substring, if such a string exists. 
If there is no matching substring, this query returns \textsf{False}.
  \item \sloppy % to allow line breaking
  \emph{Random Jumbled-Indexing} (RJI):
given a Parikh vector,
$\psi = (f_1, \dots, f_\sigma)$, this query
returns, uniformly at random, one of the indices of a 
substring with Parikh vector $\psi$ if such a substring exists in $S$.
If there is no such substring, this query returns \textsf{False}.
\end{itemize}

Unfortunately,
for the AJI variant, 
there are some strings that cannot be fully reconstructed, but this is admittedly
not obvious. In fact, the unreconstructability characterization of \cite{DBLP:conf/isit/AcharyaDMOP14,DBLP:journals/siamdm/AcharyaDMOP15} fails for AJI queries, because the symmetry property used in their construction of pairwise ``equicomposable'' strings inherently yields matching substrings with symmetric (e.g. different) positions in $S$.

Nevertheless, we
give a construction of an infinite family of pairwise undistinguishable strings, i.e. two strings
such that, for every possible query, there exists an answer 
(positive or negative) that is common to both strings.
Clearly, the adversarial strategy is to output these common answers when
given either of these strings.
In particular, for all $b\ge 1$, consider
the two binary strings of length $4 b + 14$  given below, which differ only 
in the middle section, 
consisting of \texttt{01} in the first string and \texttt{10} in the second:
  \begin{center}
    $S_1 = $ \texttt{101101(10)$^b$01(10)$^b$010010}\\
    $S_2 = $ \texttt{101101(10)$^b$10(10)$^b$010010}
  \end{center}%
\begin{restatable}{theorem}{thmAJI}
\label{thm:aji}
The strings $S_1$ and $S_2$ cannot be distinguished using AJI queries, for $b\ge 1$.
\end{restatable}

\begin{proof}
Let $n = 4 b + 14$ be the size of the strings. We refer to responses that would be common to
both $S_1$ and $S_2$ as \emph{helpless} answers. 
  Let us think of a positive answer $i$ to a query $(k,l)$ in terms of the space occupied by its matching substring, denoted $\langle i,\ i+k+l-1\rangle$. We note that an answer that does not span the middle section or that spans it in its entirety must be helpless.

  Notice that the first half of either string is the symmetric complement of the second half. This implies the following: (i) an answer $\langle i,\ j\rangle$ to a query $(k,l)$ exists if and only if an answer $\langle n-j+1,\ n-i+1\rangle$ exists for the query $(l,k)$ and (ii) an answer is negative to $(k,l)$ if and only if an answer is negative to $(l,k)$. Therefore, we can restrict ourselves to queries of the form $(k,k+c)$, where $c\ge 0$. We break this down to the following cases:

  \begin{enumerate}
    \item \textbf{Queries of type} $(k,k)$.

    We say that an answer is \emph{$k$-centered} if it is of the type $\langle n/2-(k-1),\ n/2+1+(k-1)\rangle$. Since any $k$-centered answer contains the middle section, it must be helpless. Thus, it is enough to show, by induction, that all queries $(k,k)$ have $k$-centered answers. Clearly, this holds for the base case $(1,1)$, so let us assume that there exists a $(k-1)$-centered answer $a$ to the query $(k-1,k-1)$. Then, because the first half of either string is the symmetric complement of the second half, the letters preceding and succeeding $a$ must be the complement of each other. Thus, the $k$-centered answer must be valid for the query $(k,k)$.

    \item \textbf{Queries of type} $(k,k+1)$.
    
    Take the $k$-centered answer and either extend it with one letter to the left, or one letter to the right. Exactly one of these options is a valid answer to $(k,k+1)$ (by the symmetric-complement property of the strings) and either are helpless, since they span the middle section.

    \item \textbf{Queries of type} $(k,k+2)$.

    Consider, as a base case, the answer $\langle 2,\ 3\rangle$ to the query $(0,2)$. Clearly, it is a helpless answer. Given that the letters at positions $4+2j$ and $5+2j$ are complements of each other, $\langle2,\ 3+2j\rangle$ is a valid answer to the query $(j,j+2)$, for all $0\le j\le b+2$. For greater values of $j$, the answer is helpless regardless, since it corresponds to a substring of length greater than $2b+7$, half of the string length and, therefore, it spans the middle section.

    \item \textbf{Queries of type $(k,k+c)$, for $c\ge 3$}.
    
    It is enough to analyze answers that partially span the middle section (i.e. in exactly 1 letter), since otherwise answers are automatically helpless. Let $\Delta_{i}$ denote the number of \texttt{1}'s minus the number of \texttt{0}'s for the answer $\langle i,\ n/2\rangle$, with respect to $S_2$, for all $0\le i \le n/2$ (e.g. $\Delta_{n/2}=1$, corresponds to the first letter in the middle). A simple passage from right to left, for increasing values of $i$, reveals that there exist no value of $i$ for which $\Delta_i=4$, so we do not need to handle the case $c\ge 4$. Moreover, the only values of $i$ for which $\Delta_i=3$ are $i=2$ and $i=0$, which correspond to answers for the queries $(b+1, b+4)$ and $(b+2,b+5)$, respectively. However, these queries have helpless answers: $\langle 0,\ 2b+4\rangle$ in the former and $\langle 2,\ 2b+8\rangle$ in the latter. For the first string, a similar exercise reveals that there exist no answers that partially overlap the middle section and whose difference between the number of \texttt{1}'s and the number of \texttt{0}'s is at least 3.
  \end{enumerate}
\qed
\end{proof}

In contrast, the query variant RJI can be used to reconstruct any 
length-$n$ string, $S$, without knowing the value of $n$ in advance.
In particular, it is possible to reconstruct any length-$n$ string, $S$,
using $O(\sigma + n \log n)$ RJI queries with high probability.
Our algorithm for doing this involves a reduction to a multi-window 
coupon-collector problem.

Let $\psi_i$ be a Parikh vector that is all $0$'s except for a count of $1$ for
the letter $a_i\in \Sigma$.
Note that an RJI query using $\psi_i$
will return one of the $n_i$ 
locations in $S$ with an $a_i$ uniformly at random (if $n_i>0$).
If $n_i=0$, for any $i=1,2,\ldots,\sigma$,
we learn this fact immediately after one RJI query for $\psi_i$, so let 
us assume, w.l.o.g., that $n_i>0$, for all $i=1,2,\ldots,\sigma$, after performing
an initial $\sigma$ number of RJI queries.

Recall that in the \emph{coupon-collector} problem, a collector visits
a coupon window each day and requests a coupon from an agent, who chooses one
of $n$ coupons uniformly at random and gives it to the collector, e.g.,
see~\cite{mitzup-17}.
The expected number of days required for the collector
to get all $n$ coupons is $n H_n$, 
where $H_n$ is the $n\textsuperscript{th}$ Harmonic number.
But this assumes the collector knows when they have received all
$n$ coupons (i.e., the collector knows the value of $n$).

In a coupon-collector formulation of our reconstruction problem, we instead
have $\sigma$ coupon windows, one for each letter $a_i\in \Sigma$,
where each window $i$ has $n_i$ coupons that differ from the coupons for the other
windows, and we do not know the value of any $n_i$.
Each day the collector
must choose one of the coupon windows, $i$, and request one of its coupons
(corresponding to an RJI query for $\psi_i$),
which is chosen uniformly at random from the $n_i$ coupons for window $i$.
We are interested in a strategy and analysis for the collector to
collect all $n=n_1+n_2+\cdots+n_\sigma$ coupons, with high probability (i.e.,
with probability at least $1-1/n$).

Note that although we do not know the value of any $n_i$,
we can nonetheless test whether the collector has
collected all $n$ coupons.
In particular,
suppose we have received RJI responses for all indices, $1,2,\ldots,n$, 
for letters in $S$, and let $n_i$ be the number of $a_i$'s we have found so far.
Let $\psi'=(n_1,n_2,\ldots,n_\sigma)$, and let $\psi'_i$ be equal to
$\psi'$ except that we increment $n_i$ by $1$.
If an RJI query for each $\psi'_i$ returns \textsf{False}, then we know we have
fully reconstructed $S$.
Thus, if $n=1$, then we can determine this and $S$ after $2\sigma$ RJI queries,
so let us assume that $n\ge 2$.
Further, we can assume we have a bound, $N\ge2$, which is at least
$n$ and at most twice $n$,
by a simple doubling strategy, where we double $N$ any time a test for $n$ fails
and we set $N$ equal to any RJI query response that is larger than $N$.
Therefore,
the remaining problem is to solve
the multi-window coupon-collector problem.

Our strategy for the multi-window coupon-collector problem is simply to visit
the coupon windows in phases, so that in phase~$i$ we repeatedly visit window~$i$
until we are confident we have all of its $n_i$ coupons,
for which the following lemma will prove useful.

\begin{restatable}{lemma}{lemmaCoupon}
\label{lem:coupon}
Let $T_i$ be the number of trips to window $i$ needed 
to collect all its $n_i\ge1$ coupons.
Then, for any real number $\beta$:
\[
\Pr\left(T_i>\beta n_i\ln N\right) \le \frac{n_i}{N^{\beta}} .
\]
\end{restatable}

\begin{proof}
Adapting a proof from~\cite{wiki:coupon},
let $Z_{j,r}$ denote the event that the $j$-th coupon was not picked
in the first $r$ trips to window $i$.
Then
\[
\Pr\left(Z_{j,r}\right) = \left(1 - \frac{1}{n_i}\right)^r \le e^{-r/n_i} .
\]
Thus, for $r=\beta n_i\ln N$, we have
$\Pr(Z_{j,r})\le e^{-(\beta n_i\ln N)/n_i} = N^{-\beta}$.
Therefore,
by a union bound,
\[
\Pr\left(T> \beta n_i \ln N\right) = 
\Pr\left(\bigcup_j Z_{j,\beta n_i\ln N}\right) 
\le n_i \cdot
\Pr\left(Z_{1,\beta n_i\ln N}\right) \le \frac{n_i}{N^{\beta}}.
\]
\qed
\end{proof}

Our strategy, then, 
is to 
let $\beta\ge 2$ be constant, and in phase~$i$,
implement a doubling strategy where we perform
$\beta N_i\log N$ RJI queries for $\psi_i$, such that $N_i$ is an upper bound
estimate for $n_i$, which we double each time we get more than $N_i$ distinct
responses to our queries in this phase.
So by the end of the phase~$i$, $n_i \le N_i\le 2n_i$.
This gives us:

\begin{restatable}{theorem}{thmRJI}
\label{thm:rji}
A string, $S$, of unknown size, $n$, can be reconstructed using
$O(\sigma + n\log n)$ RJI queries, with high probability.
\end{restatable}

\begin{proof}
After an initial $O(\sigma)$ queries to determine which letters 
from $\Sigma$ appear in $S$,
the total number of remaining queries performed by our method is at most
\[
2\sum_{i=1}^\sigma 2\beta N_i \ln N
=
4\sum_{i=1}^\sigma \beta N_i \ln N,
\]
by the doubling strategy applied to each letter, $a_i\in\Sigma$, and
then globally for $N$.
Further,
\[
\sum_{i=1}^\sigma \beta N_i \ln N \le
\sum_{i=1}^\sigma 2\beta n_i \ln N
\le
\sum_{i=1}^\sigma 3\beta n_i \ln n = 3\beta n \ln n,
\]
with probability at least
\[
1 - \sum_{i=1}^\sigma \frac{n_i}{N^\beta} 
=
1 - \frac{\sum_{i=1}^\sigma n_i}{N^\beta} 
= 1-\frac{n}{N^\beta} 
\ge 1-\frac{1}{n},
\]
by \cref{lem:coupon}, since $\beta\ge 2$.
\qed
\end{proof}

\section{Conclusion and Open Questions}

We have studied the reconstruction of strings under the following
settings, by giving efficient reconstruction algorithms and proving
lower bounds: (i) periodic strings of known and unknown sizes, with
and without mismatch errors, using substring queries; (ii) periodic
strings of known and unknown sizes, using subsequence queries and
(iii) general strings, using variations of jumbled-indexing queries.
For the non-optimal algorithms given here, it would be nice to know
whether there exist matching lower bounds, or whether there exist
faster algorithms.

Regarding corrupted periodic strings, different applications suffer from different types of
corruption. In particular, the following error metrics have been
considered in the literature: Pseudo-local metrics such as swap
distance~\cite{AALLL:00} or Interchange (Cayley) distance~\cite{aelps:12}; and the Levenshtein edit distance~\cite{L-66}. It
would be interesting to see whether our reconstruction algorithms can
be adapted to these more general error distances.
 
The next step is to reconstruct strings that have more complex syntactic regularities than periods, such as {\em covers}~\cite{AIF-90}.
 A length $m$ substring $C$ of a string $T$ of length $n$, is said to be a \emph{cover} of $T$,
  if $n>m$ and every letter of $T$ lies within some occurrence of   $C$. We would like to efficiently reconstruct 
a coverable string, without knowing its cover a-priori. 

Data compression schemes such as, Lempel-Ziv~\cite{DBLP:journals/tit/ZivL77,ZL-78} are
known to compress any stationary and ergodic source down
to the entropy rate of the source per source symbol, provided
the input source sequence is sufficiently long. These schemes 
rely heavily on encoding repeated substrings by their starting index and
length. In this sense, a periodic string is highly compressible. We
would like to extend our ideas to reconstruct a general string in time
proportional to its LZ compression.

The type of query used for reconstruction is a key factor in the
reconstruction complexity. Much as the error distance, the query type
is also application-dependent. A reasonable query type is the {\em
  less than matching}. Let $S_1$ and $S_2$ be strings of length
$n$ over an ordered alphabet. We say that $S_1$ is {\em less than}
$S_2$ if $S_1[i] < S_2,\ \forall i=1,\dots,n$. Other matchings that have
been researched in the literature,
are the {\em order preserving matching}~\cite{OrderPreservingSuffixTreeCrochemore-16,cnps:13,kanps:14}, and
the {\em parameterized matching}~\cite{bak:96,bak:97}. In 
the order preserving matching, we say that two strings match if the
relative order of their elements is the same, for example $1,2,3,2,1$
matches such strings as $1, 100, 101, 100, 1$, or $56,61,366,61,56$,
i.e., any string where the fist element is smaller than the second,
which is smaller than the third, where the fourth is equal to the
second, and the fifth equals the first. Two equal-length strings
$S_1,S_2$ over alphabet $\Sigma$ are said to {\em parameterize match},
if there is a bijection $f:\Sigma \rightarrow \Sigma$ such that
$S_1=f(S_2)$. Using these more powerful queries, can we reconstruct a
string more efficiently?

Finally, given the impossibility result on reconstructing strings using Adversarial Jumbled-Indexing queries, it would be interesting to know whether there exists an efficient algorithm that enumerates all of the undistinguishable strings.

\subsection*{Acknowledgments}

This research was funded in part by the U.S. National Science Foundation under grant 1815073. Amihood Amir was partly supported by BSF grant 2018141 and ISF grant 1475-18.

\clearpage
\bibliographystyle{splncs04}
\bibliography{ref,refs,paper}

\end{document}